\documentclass[12pt]{article} 
\DeclareMathAlphabet{\mathbbmsl}{U}{bbm}{m}{sl}
\usepackage{graphicx}
\usepackage{enumerate} 
\usepackage{algorithm}
\usepackage{fullpage} 
\usepackage{color}
\usepackage{palatino} 
\usepackage{amsmath}
\usepackage{amsthm}
\usepackage{amssymb}

\newcommand{\T}{\mathbbmsl{T}}
\newcommand{\F}{\mathbbmsl{F}}
\newcommand{\G}{\mathbbmsl{G}}
\newcommand{\II}{\mathbbmsl{I}}
\newcommand{\SSS}{\mathbbmsl{S}}
\newcommand{\U}{\mathbbmsl{U}}
\newcommand{\C}{\mathbbmsl{C}}
\newcommand{\V}{\textit{Node}}
\newcommand{\vv}{\textit{node}}
\newcommand{\N}{\textit{Nbr}}
\newcommand{\nn}{\textit{nbr}}
\newcommand{\D}{\textit{Offspring}}
\newcommand{\dd}{\textit{offspring}}
\newcommand{\xx}{\textit{hook}}

\newtheorem{claim}{Claim}
\newcounter{pcounter}
\newtheorem{lemma}{Lemma}[section]
\newtheorem{theorem}{Theorem}[section]
\newtheorem{corollary}{Corollary}[section]
\title{Linear-Time Compression of Bounded-Genus Graphs into
Information-Theoretically Optimal Number of Bits\thanks{Accepted to {\em SIAM Journal on Computing}. A 
preliminary version appeared in SODA~\cite{Lu02}.}}

\author{Hsueh-I Lu\thanks{Department of Computer Science and
    Information Engineering, National Taiwan University.  Address: 1
    Roosevelt Road, Section 4, Taipei 106, Taiwan, ROC.  Web:
    www.csie.ntu.edu.tw/$\sim$hil/.  Email: hil@csie.ntu.edu.tw.
    Research supported in part by NSC grant 101--2221--E--002--062--MY3.
    The author also holds joint appointments from the Graduate
    Institute of Networking and Multimedia and the Graduate Institute
    of Biomedical Electronics and Bioinformatics, National Taiwan
    University.}
}

\date{January 11, 2014}

\begin{document}
\maketitle 
\begin{abstract}
A {\em compression scheme} $A$ for a class $\G$ of graphs consists of
an encoding algorithm $\textit{Encode}_A$ that computes a binary
string $\textit{Code}_A(G)$ for any given graph $G$ in $\G$ and a
decoding algorithm $\textit{Decode}_A$ that recovers $G$ from
$\textit{Code}_A(G)$.  A compression scheme $A$ for $\G$ is {\em
  optimal} if both $\textit{Encode}_A$ and $\textit{Decode}_A$ run in
linear time and the number of bits of $\textit{Code}_A(G)$ for any
$n$-node graph $G$ in $\G$ is information-theoretically optimal to
within lower-order terms.  Trees and plane triangulations were the
only known nontrivial graph classes that admit optimal compression
schemes.  Based upon Goodrich's separator decomposition for planar
graphs and Djidjev and Venkatesan's planarizers for bounded-genus
graphs, we give an optimal compression scheme for any hereditary
(i.e., closed under taking subgraphs) class $\G$ under the premise
that any $n$-node graph of $\G$ to be encoded comes with a
genus-$o(\frac{n}{\log^2 n})$ embedding.  By Mohar's linear-time
algorithm that embeds a bounded-genus graph on a genus-$O(1)$ surface,
our result implies that any hereditary class of genus-$O(1)$ graphs
admits an optimal compression scheme.
For instance, our result yields the first-known optimal compression
schemes for planar graphs, plane graphs, graphs embedded on genus-$1$
surfaces, graphs with genus $2$ or less, $3$-colorable directed plane
graphs, $4$-outerplanar graphs, and forests with degree at most $5$.
For non-hereditary graph classes, we also give a methodology for
obtaining optimal compression schemes. From this methodology, we give
the first known optimal compression schemes for triangulations of
genus-$O(1)$ surfaces and floorplans.
\end{abstract}

\section{Introduction}
Compact representation of graphs are fundamentally important and
useful in many applications, including representing the meshes in
finite-element analysis, terrain models of GIS, and 3D models of
graphics~\cite{Rossignac99,Rossignac-edgebreaker,RossignacSS01,TaubinR98,SnoeyinkVK97,LopesTRSS02,SzymczakKR01,IsenburgS01}, VLSI
design~\cite{SahaS09,KrivogradTZ08}, designing compact routing tables
of computer
networks~\cite{Thorup04,GavoilleH99,Lu10,Peleg00,Gavoille00,ThorupZ01,AbrahamMR09,Chechik11,AgarwalGH11,GavoilleS11},
and compressing the link structure of the
Internet~\cite{BroderKMPRSTW00,AdlerM01,SuelY01,AsanoMN09,AnhM10,ClaudeN10}.
Let $\G$ be a class of graphs.  Let $\textit{num}(\G,n)$ denote the
number of distinct $n$-node graphs in $\G$.  The
information-theoretically optimal number of bits to encode an $n$-node
graph in $\G$ is $\lceil\log\textit{num}(\G,n)\rceil$.\footnote{All
  logarithms throughout the paper are to the base of two.}  For
instance, if $\G$ is the class of rooted trees, then
$\textit{num}(\G,n)\approx \frac{2^{2n}}{n^{3/2}}$ and
$\log\textit{num}(\G,n)=2n-O(\log n)$; if $\G$ is the class of plane
triangulations, then $\log\textit{num}(\G,n)=\log \frac{256}{27} n +
o(n)\approx 3.2451n+o(n)$~\cite{Tutte62}.
A {\em compression scheme} $A$ for $\G$ consists of an encoding
algorithm $\textit{Encode}_A$ that computes a binary string
$\textit{Code}_A(G)$ for any given graph $G$ in $\G$ and a decoding
algorithm $\textit{Decode}_A$ that recovers graph $G$ from
$\textit{Code}_A(G)$.  A compression scheme $A$ for a graph class $\G$
with $\log \textit{num}(\G,n)=O(n)$ is {\em optimal} if the following
three conditions hold.
\begin{enumerate}[\em {Condition}~C1:]
\addtolength{\itemsep}{-0.5\baselineskip}
\item 
\label{condition:c1}
The running time of algorithm $\textit{Encode}_A(G)$ is linear in the size of
$G$.

\item 
\label{condition:c2}
The running time of algorithm $\textit{Decode}_A(\textit{Code}_A(G))$ is linear
in the bit count of $\textit{Code}_A(G)$.

\item 
\label{condition:c3}
For all positive constants $\beta$ with $\log\textit{num}(\G,n)\leq
\beta n+o(n)$, the bit count of $\textit{Code}_A(G)$ for an $n$-node
graph $G$ in $\G$ is no more than $\beta n+o(n)$.
\end{enumerate}
Condition~C\ref{condition:c3} basically says that the bit count of
$\textit{Code}_A(G)$ is information-theoretically optimal to within
lower-order terms.  Although there has been considerable work on
compression schemes, trees (see
e.g.,~\cite{MunroR01,Jacobson89,LuY08,BenoitDMRRR05}) and plane
triangulations~\cite{PoulalhonS06} were the only known nontrivial
graph classes that admit optimal compression schemes.
A graph class is {\em hereditary} if it is closed under taking
subgraphs.  Below is the main result of the paper.
\begin{theorem}
\label{theorem:theorem1}
Any hereditary class $\G$ of graphs with 
$\log \textit{num}(\G,n)=O(n)$ admits an optimal compression scheme, as long
as each input $n$-node graph in $\G$ to be encoded comes with a
genus-$o(\frac{n}{\log^2 n})$ embedding.
\end{theorem}
\noindent
By Theorem~\ref{theorem:theorem1} and Mohar's linear-time genus-$O(1)$
embedding algorithm for genus-$O(1)$ graphs~\cite{Mohar99,KawarabayashiMR08}
(see Lemma~\ref{lemma:mohar}), any hereditary class of genus-$O(1)$
graphs admits an optimal compression scheme.
For instance, our result yields the first-known optimal compression
schemes for planar graphs, plane graphs, graphs embedded on genus-$1$
surfaces, graphs with genus $2$ or less, $3$-colorable directed plane
graphs, $4$-outerplanar graphs, and forests with degree at most $5$.
For non-hereditary graph classes, we also give an extension (see
Corollary~\ref{corollary:corollary1}) of
Theorem~\ref{theorem:theorem1}.  As summarized in the following
theorem, we show 
two classes of genus-$O(1)$ graphs whose optimal
compression schemes are obtainable via this extension, where the class
of floorplans is defined in related work below.
\begin{theorem}
\label{theorem:theorem2}
The following
classes of graphs admit optimal compression schemes:
\begin{enumerate}
\addtolength{\itemsep}{-0.5\baselineskip}
\item
\label{thm2:triangulation}
Triangulations of a genus-$g$ surface for any integral constant $g$.
\item 
\label{thm2:floorplan}
Floorplans.
\end{enumerate}
\end{theorem}

\paragraph{Technical overview}
The kernel of the proof of Theorem~\ref{theorem:theorem1} is a
linear-time disjoint partition $G_0,\ldots,G_p$ of an $n$-node graph
$G$ embedded on a genus-$o(\frac{n}{\log^2 n})$
surface.\footnote{Precisely, the disjoint partition $G_0,\ldots,G_p$
  of the edges of the embedded graph $G$ in the proof of
  Theorem~\ref{theorem:theorem1} is $G[V_0], G(V_1), \ldots,G(V_p)$,
  where $[V_0,\ldots,V_p]$ is both (i) a $1$-separation $\SSS_1$ of an
  arbitrary triangulation $\Delta$ of $G$ and (ii) a refinement of the
  $0$-separation $\SSS_0=[\varnothing,\V(\Delta)]$ of $\Delta$.}  Let
$\text{poly}(n)$ denote $O(n^{O(1)})$.  Based upon Goodrich's
separator decomposition of planar graphs~\cite{Goodrich95} and Djidjev
and Venkatesan's planarizer~\cite{DjidjevV95}, partition
$G_0,\ldots,G_p$ satisfies the following conditions, where $n_i$ is
the number of nodes of $G_i$ and $d_i$ is the number of times that the
nodes of $G_i$ are duplicated in some $G_j$ with $j\ne i$:\footnote{As
  a matter of fact, in our construction, all duplicated nodes of $G_i$
  with $i\geq 1$ belong to $G_0$.}  (a)~$n_0 = o(\frac{n}{\log n})$,
(b)~$n_i=\text{poly}(\log n)$ holds for each $i=1,2,\ldots,p$,
(c)~$\sum_{i=1}^p d_i=o(\frac{n}{\log n})$, and~(d)~$\sum_{i=0}^p
n_i=n+o(\frac{n}{\log n})$.  By Condition~(a), $G_0$ can be encoded in
$o(n)$ bits. By Conditions~(b) and (c), the information required to
recover $G$ from $G_0,G_1,\ldots,G_p$ can be encoded into $o(n)$ bits
(see Lemma~\ref{lemma:recovery}).  By Condition~(d), we have
$\log\textit{num}(\G,n)\leq o(n)+\sum_{i=1}^p
\log\textit{num}(\G,n_i)$.  Therefore, the disjoint partition reduces
the problem of encoding an $n$-node graph in $\G$ to the problem of
encoding a $\text{poly}(\log n)$-node graph in $\G$.  Applying such a
reduction for one more level, it remains to encode a
$\text{poly}(\log\log n)$-node graph in $\G$ into an
information-theoretically optimal number of bits, which can be
resolved by the standard technique~(see,
e.g.,~\cite{HeKL00,MunroR01,PettieR02}) of precomputation tables (see
Lemma~\ref{lemma:basis}).

\paragraph{Related work}
The compression scheme of Tur\'{a}n~\cite{turan84} encodes an $n$-node
plane graph that may have self-loops into $12n$ bits.\footnote{For
  brevity, we omit all lower-order terms of bit counts in our
  discussion of related work.}  Keeler and Westbrook~\cite{KeelerW95}
improved this bit count to $10.74n$.  They also gave compression
schemes for several families of plane graphs. In particular, they used
$4.62n$ bits for plane triangulation, and $9n$ bits for connected
plane graphs free of self-loops and degree-one nodes.  For plane
triangulations, He et~al.~\cite{HeKL99} improved the bit count to
$4n$.  For triconnected plane graphs, He et~al.~\cite{HeKL99} also
improved the bit count to at most $8.585n$ bits.  This bit count was
later reduced to at most $\frac{9\log_2 3}{2}n\approx 7.134n$ by
Chuang et~al.~\cite{ChuangGHKL98}.  For any given $n$-node graph $G$
embedded on a genus-$g$ surface, Deo and Litow~\cite{DeoL98} showed an
an $O(ng)$-bit
encoding for $G$. These compression schemes all take linear time for
encoding and decoding, but Condition~C\ref{condition:c3} does not hold
for them.  The compression schemes of He et~al.~\cite{HeKL00}
(respectively, Blelloch et~al.~\cite{BlellochF10}) for planar graphs,
plane graphs, and plane triangulations (respectively, separable
graphs) satisfies Condition~C\ref{condition:c3}, but their encoding
algorithms require $\Omega(n\log n)$ time on $n$-node graphs.

\begin{figure}[t]
\centerline{\input{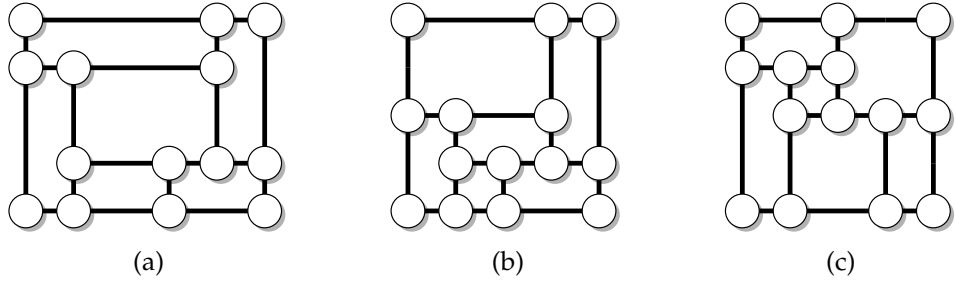}}
\caption{Three floorplans with $14$ nodes, $6$ internal faces, and
  $19$ edges. Floorplans~(a) and~(b) are equivalent, floorplans~(b)
  and~(c) are not equivalent.}
\label{figure:floorplan}
\end{figure}

Floorplanning is a fundamental issue in circuit
layout~\cite{YaoCCG03,GuoTCY01,MoffittRMP08,LinC05,Kajitani08,Young08,FengM06,BanerjeeSB09,ChenC06,LiSC10,DasguptaS01,LaiL88,TTSS91,MMH82,SahaS09,AgnihotriOM10}.
Motivated by VLSI physical design, various representations of
floorplans were proposed~\cite{ZhuangSJK03,ZhaoLKS04,FujimakiT07}.
Designing a floorplan to meet a certain criterion is NP-complete in
general~\cite{L83,H88,WKC88}, so heuristic techniques such as
simulated annealing~\cite{WL89,WLL88,ChenC06} are practically
useful. The length of the encoding affects the size of the search
space.  A {\em floorplan}, which is also known as {\em rectangular
  drawing}, is a division of a rectangle into rectangular faces using
horizontal and vertical line segments.  Two floorplans are {\em
  equivalent} if they have the same adjacency relations and relative
positions among the nodes.  For instance,
Figure~\ref{figure:floorplan} shows three floorplans: Floorplans (a)
and (b) are equivalent. Floorplans (b) and (c) are not equivalent.
Let $G$ be the input $n$-node floorplan.  Under the conventional
assumption that each node of $G$, other than the four corner nodes,
has exactly three neighbors (see, e.g.,~\cite{He99,YS93}), one can
verify that $G$ has $0.5n$ faces and $1.5n-2$ edges.  Yamanaka and
Nakano~\cite{YN06} showed how to encode $G$ into $2.5n$ bits.
Chuang~\cite{Chuang08} reduced the bit count to $2.293n$.  Takahashi
et~al.~\cite{TakahashiFI09} further reduced bit count to $2n$.  All
these compression schemes for floorplans satisfy
Conditions~C\ref{condition:c1} and~C\ref{condition:c2}, but not
Condition~C\ref{condition:c3}.  Takahashi et~al.~\cite{TakahashiFI09}
also showed that the number of distinct $n$-node floorplans is no more
than $3.375^{n+o(n)}\approx 2^{1.755n+o(n)}$. Therefore, our
Theorem~\ref{theorem:theorem2}(\ref{thm2:floorplan}) encodes an $n$-node
floorplan into at most $1.755n$ bits.

For applications that require query support,
Jacobson~\cite{Jacobson89} gave a $\Theta(n)$-bit encoding for a
connected and simple planar graph $G$ that supports traversal in
$\Theta(\log{n})$ time per node visited. Munro and
Raman~\cite{MunroR97} improved this result and gave schemes to encode
binary trees, rooted ordered trees, and planar graphs. For a general
$n$-node $m$-edge planar graph $G$, they used $2m+8n$ bits while
supporting adjacency and degree queries in $O(1)$ time. Chuang
et~al.~\cite{ChuangGHKL98} reduced this bit count to
$2m+(5+\frac{1}{k})n$ for any constant $k>0$ with the same query
support. The bit count can be further reduced if only $O(1)$-time
adjacency queries are supported, or if $G$ is simple, triconnected or
triangulated~\cite{ChuangGHKL98}.  Chiang et~al.~\cite{ChiangLL05}
reduced the number of bits to $2m+2n$.  Yamanaka and
Nakano~\cite{YamanakaN10} showed a $6n$-bit encoding for plane
triangulations with query support.  The succinct encodings of Blandford
et~al.~\cite{BlandfordBK03} and Blelloch et~al.~\cite{BlellochF10} for
separable graphs support queries.  Yamanaka et~al.~\cite{YamanakaN08}
also gave a compression scheme for floorplans with query support.
For labeled planar graphs, Itai and Rodeh~\cite{ItaiR82} gave an
encoding of $\frac{3}{2} n \log n$ bits.  For unlabeled general
graphs, Naor~\cite{Naor90} gave an encoding of ${\frac{1}{2}}n^2$
bits.
For certain graph families, Kannan et~al.~\cite{KNR92} gave schemes
that encode each node with $O(\log n)$ bits and support $O(\log
n)$-time testing of adjacency between two nodes. 
Galperin and Wigderson~\cite{GW83}
and Papadimitriou and Yannakakis~\cite{PH86.encode} investigated
complexity issues arising from encoding a graph by a small circuit
that computes its adjacency matrix.
Related work on various versions of succinct graph representations can
be found
in~\cite{MunroRS01,ArikatiMZ97,FederM95,GrossiL98,GavoilleH:k-page,Patrascu08,SadakaneN10,FarzanRR09,FarzanM08a,FarzanM08b,BarbayAHM07,kaot93.joa}
and the references therein.

\paragraph{Outline}
The rest of the paper is organized as follows.
Section~\ref{section:prelim} gives the preliminaries.
Section~\ref{section:separation} shows our algorithm for computing
graph separations.  Section~\ref{section:compression} gives our
optimal compression scheme for hereditary graph classes.
Section~\ref{section:extend} shows a methodology for obtaining optimal
compression schemes for non-hereditary graph classes and applies this
methodology on triangulations of genus-$O(1)$ graphs
and floorplans.  Section~\ref{section:conclude} concludes the paper
with a couple of open questions.

\section{Preliminaries}
\label{section:prelim}

\subsection{Segmentation prefix}
Let $\|X\|$ denote the number of bits of binary string $X$.  A binary
string $X_0$ is a {\em segmentation prefix} of binary strings
$X_1,\ldots,X_d$ if (a) it takes $O(\sum_{i=1}^d \|X_i\|)$ time to
compute $X_0$ from $X_1,\ldots,X_d$ and (b) given the concatenation of
$X_0,X_1,\ldots,X_d$, it takes $O(\sum_{i=0}^d \|X_i\|)$ time to
recover all $X_i$ with $1\leq i\leq d$.

\begin{lemma}[See, e.g.,~\cite{BellCW90,Elias75}]
\label{lemma:concat}
Any binary strings $X_1,\ldots,X_d$ with $d=O(1)$ have a segmentation
prefix with $O(\log \sum_{i=1}^d \|X_i\|)$ bits.
\end{lemma}

\begin{lemma}
\label{lemma:unary-concat}
Any binary strings $X_1,\ldots,X_d$ have an $O(\min\{m,\ d\log
m\})$-bit segmentation prefix, where $m=\|X_1\|+\cdots+\|X_d\|$.
\end{lemma}

\begin{proof}
Let $X$ be the concatenation of $X_1,\ldots,X_d$.  If $m\leq d\log m$,
let $X'$ be the $m$-bit binary string with exactly $d$ copies of
$1$-bits such that the $j$-th bit of $X'$ is $1$ if and only if
$j=\|X_1\|+\cdots+\|X_i\|$ holds for some $i=1,\ldots,d$.  Otherwise,
let $X'$ store the $O(\log m)$-bit numbers $\|X_1\|+\cdots+\|X_i\|$
for all $i=1,\ldots,d$.  Let $X'_0$ be the segmentation prefix of $X'$
and $X$ as ensured by Lemma~\ref{lemma:concat}.  The concatenation of
$X'_0$ and $X'$ is a segmentation prefix $X_0$ of $X_1,\ldots,X_d$
with $O(\min\{m,\ d\log m\})$ bits.  The lemma is proved.
\end{proof}

\noindent
For the rest of the paper, let $X_1\circ \cdots \circ X_d$ denote the
concatenation of $X_0,X_1,\ldots,X_d$, where $X_0$ is the segmentation
prefix of $X_1,\ldots,X_d$ as ensured by
Lemma~\ref{lemma:unary-concat}.

\subsection{Precomputation table}

Unless clearly stated otherwise, all graphs throughout the paper are
simple, i.e., having no multiple edges or self-loops.  
Let $|S|$ denote the cardinality of set $S$.  Let $\V(G)$ consist of
the nodes in graph $G$ and let $\vv(G)=|\V(G)|$.  For any subset $V$
of $\V(G)$, let $G[V]$ denote the subgraph of $G$ induced by~$V$ and
let $G\setminus V$ denote the subgraph of $G$ obtained by deleting $V$
and their incident edges.  Two disjoint subsets $V$ and $V'$ of
$\V(G)$ are {\em adjacent} in $G$ if there is an edge $(v,v')$ of $G$
with $v\in V$ and $v'\in V'$.  For any subset $V$ of $\V(G)$, let
$\N_G(V)$ consist of the nodes in $\V(G)\setminus V$ that are adjacent
to $V$ in $G$ and let $\nn_G(V)=|\N_G(V)|$.  A {\em connected
  component} of graph $G$ is a maximal subset $C$ of $\V(G)$ such that
$G[C]$ is connected.

\begin{lemma}
\label{lemma:basis}
Let $\G$ be a graph class satisfying $\log \textit{num}(\G,n)=O(n)$.
Given positive integers $\ell$ and~$n$ with $\ell=\text{poly}(\log\log
n)$, it takes overall $o(n)$ time to compute (i) a labeling
$\textit{Label}(H)$ and a $\lceil\log
\textit{num}(\G,\vv(H))\rceil$-bit binary string $\textit{Optcode}(H)$
for each distinct graph $H\in\G$ with at most $\ell$ nodes and (ii) an
$o(n)$-bit string $\textit{Table}(\G,\ell)$ such that the following
statements hold.
\begin{enumerate}
\addtolength{\itemsep}{-0.5\baselineskip}
\item 
\label{basis:1}
Given any graph $H\in \G$ with $\vv(H)\leq \ell$, it takes $O(\vv(H))$
time to obtain $\textit{Optcode}(H)$ and $\textit{Label}(H)$ from
$\textit{Table}(\G,\ell)$.

\item  
\label{basis:2}
Given $\textit{Optcode}(H)$ for any graph $H\in \G$ with $\vv(H)\leq
\ell$, it takes $O(\vv(H))$ time to obtain $H$ and $\textit{Label}(H)$
from $\textit{Table}(\G,\ell)$.
\end{enumerate}
\end{lemma}

\begin{proof}
Straightforward by $O(1)^{\text{poly}(\ell)}=o(n)$.
\end{proof}

\subsection{Separator decomposition of planar graphs}

\begin{figure}[t]
\centerline{\input{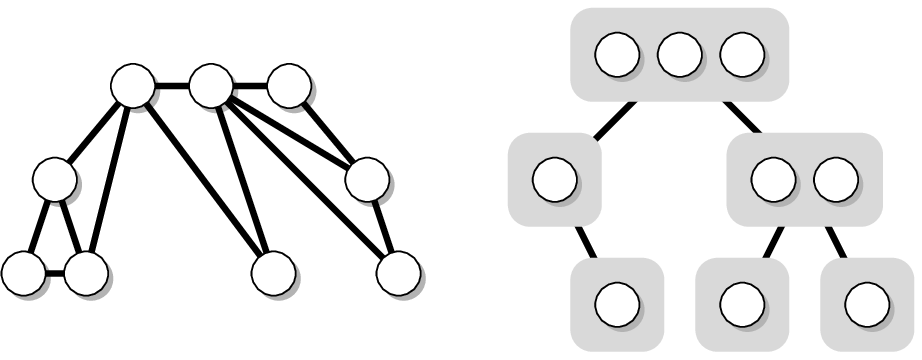}}
\caption{(a) A $9$-node plane graph $G$.
 (b) A separator decomposition $\T$ of $G$.}
\label{figure:goodrich}
\end{figure}

Sets $S_1,\ldots,S_d$ form a {\em disjoint partition} of set $S$ if
$S_1,\ldots,S_d$ are pairwise disjoint and $S=S_1\cup \cdots\cup S_d$.
A subset $S$ of $\V(G)$ is a {\em separator} of graph $G$ with respect
to $S_1$ and $S_2$ if (1) $S$, $S_1$, and $S_2$ form a disjoint
partition of $\V(G)$, (2) $S_1$ and $S_2$ are not adjacent in $G$, (3)
$|S|=O(\vv(G)^{1/2})$, and (4)
$\max\{|S_1|,|S_2|\}\leq\frac{2}{3}\cdot\vv(G)$.  A {\em separator
  decomposition}~\cite{BhattL84} of $G$ is a rooted binary tree $\T$
on a disjoint partition of $\V(G)$ such that the following two
statements hold, where ``nodes'' specify elements of $\V(G)$ and
``vertices'' specify elements of $\V(\T)$. Statement~1: Each leaf
vertex of $\T$ consists of a single node of $G$. Statement~2: Each
internal vertex $S$ of $\T$ is a separator of $G[\D(S)]$ with respect
to $\D(S_1)$ and $\D(S_2)$, where $S_1$ and $S_2$ are the child
vertices of $S$ in $\T$ and $\D(S)$ (respectively, $\D(S_1)$ and
$\D(S_2)$) is the union of all the vertices in the subtree of $\T$
rooted at $S$ (respectively, $S_1$ and $S_2$).  See
Figure~\ref{figure:goodrich} for an illustration.

\begin{lemma}[Goodrich~\cite{Goodrich95}]
\label{lemma:goodrich}
It takes $O(n)$ time to compute a separator decomposition for any given $n$-node planar graph.
\end{lemma}

\subsection{Planarizers for non-planar graphs}
The {\em genus} of a graph $G$ is the smallest integer $g$
such that $G$ can be embedded on an orientable surface with $g$
handles without edge crossings~\cite{Gross87}. For example, the genus
of a planar graph is zero.
By Euler's formula (see, e.g.,~\cite{GilbertHT84}), an $n$-node
genus-$O(n)$ graph has $O(n)$ edges.  Determining the genus of a
general graph is NP-complete~\cite{Thomassen89}, but
Mohar~\cite{Mohar99} showed that it takes linear time to determine
whether a graph is of genus $g$ for any $g=O(1)$. Mohar's algorithm is
simplified by Kawarabayashi et~al.~\cite{KawarabayashiMR08}.

\begin{lemma}[Mohar et al.~\cite{Mohar99,KawarabayashiMR08}]
\label{lemma:mohar}
It takes $O(n)$ time to compute a genus-$O(1)$ embedding for any given
$n$-node genus-$O(1)$ graph.
\end{lemma}

Gilbert et~al.~\cite{GilbertHT84} gave an $O(n+g)$-time algorithm to
compute an $O((gn)^{0.5})$-node separator of an $n$-node genus-$g$
graph, generalizing Lipton and Tarjan's classic separator theorem for
planar graphs~\cite{LiptonT79}. Our result relies on the following
planarization algorithm.

\begin{lemma}[Djidjev and Venkatesan~\cite{DjidjevV95}]
\label{lemma:2}
\label{lemma:djidjev}
Given an $n$-node graph $G$ embedded on a genus-$g$ surface, it takes
$O(n+g)$ time to compute a subset $V$ of $\V(G)$ with
$|V|=O((gn)^{0.5})$ such that $G\setminus V$ is planar.
\end{lemma}

\section{Separation and refinement}
\label{section:separation}

\begin{figure}[t]
\centerline{\input{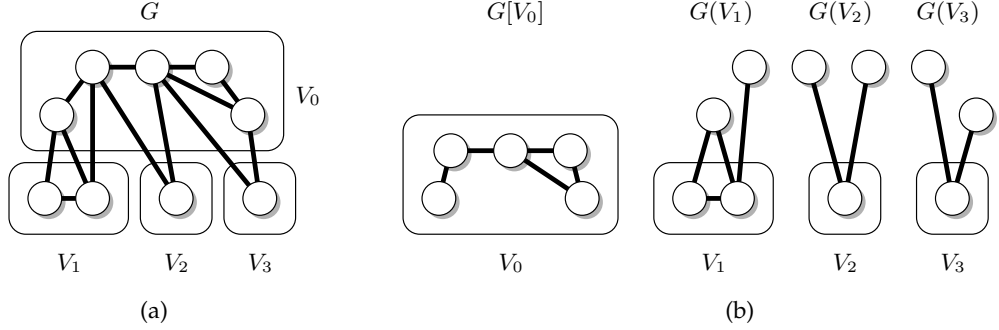}}
\caption{(a) A $9$-node plane graph with a separation
  $[V_0,\ldots,V_3]$.  (b) $G[V_0]$, $G(V_1)$, $G(V_2)$, and $G(V_3)$
  form a disjoint partition of the edges of $G$.}
\label{figure:subgraph}
\end{figure}

We say that $[V_0,\ldots,V_p]$ with $p\geq 1$ is a {\em separation} of
graph $G$ if the following properties hold.
\begin{enumerate}[\em Property~S1:]
\addtolength{\itemsep}{-0.5\baselineskip}
\item 
\label{separation:1}
$V_0,\ldots,V_p$ form a disjoint partition of $\V(G)$.
\item 
\label{separation:2}
Any two $V_i$ and $V_{i'}$ with $1\leq i\ne i'\leq p$ are not adjacent
in $G$. 
\setcounter{pcounter}{\theenumi}
\end{enumerate}
For instance, Figure~\ref{figure:subgraph}(a) shows a separation
$[V_0,V_1,V_2,V_3]$ of graph $G$ and Figure~\ref{figure:refinement}(a)
shows another separation $[U_0,U_1,U_2]$ of $G$.
For any subset $V$ of $\V(G)$, let $G(V)$ be the subgraph of $G$
induced by $V\cup \N_G(V)$ excluding the edges of $G[\N_G(V)]$.  If
$[V_0,\ldots,V_p]$ is a separation of $G$, then
$G[V_0],G(V_1),\ldots,G(V_p)$ form a disjoint partition of the edges
of $G$.  See Figures~\ref{figure:subgraph}(b)
and~\ref{figure:refinement}(b) for illustrations.
Let $\log^{(0)}n = n$.  For any positive integer $k$, let $\log^{(k)}
n=\log\ (\log^{(k-1)} n)$.  For notational brevity, for any
nonnegative integer $k$, let
\begin{displaymath}
\ell_k=\max\{1,\log^{(k)} n\}.
\end{displaymath}
For a nonnegative integer $k$, separation $[V_0,\ldots,V_p]$ of an
$n$-node graph $G$ is a {\em $k$-separation} of $G$ if the following
three properties hold.
\begin{enumerate}[\em Property~S1:]
\addtolength{\itemsep}{-0.5\baselineskip}
\setcounter{enumi}{\thepcounter}
\item 
\label{separation:3}
$|V_0|=o(\frac{n}{\ell_k})$ and $p=o(\frac{n}{\ell_k})+1$.

\item 
\label{separation:4}
$|V_i|+\nn_G(V_i)=\text{poly}(\ell_k)$ holds for each
$i=1,\ldots,p$.

\item
\label{separation:5}
$\sum_{i=1}^p\nn_G(V_i)=o(\frac{n}{\ell_k})$.
\end{enumerate}
One can verify that $[\varnothing, \V(G)]$ is a $0$-separation of
$G$.\footnote{The ``$+1$'' in Property~S\ref{separation:3} is
  redundant for $k\geq 1$. However, we need it so that $[\varnothing,
    \V(G)]$ is a $0$-separation of $G$, since $1\ne
  o(\frac{n}{\ell_0})$.}
\begin{figure}[t]
\centerline{\input{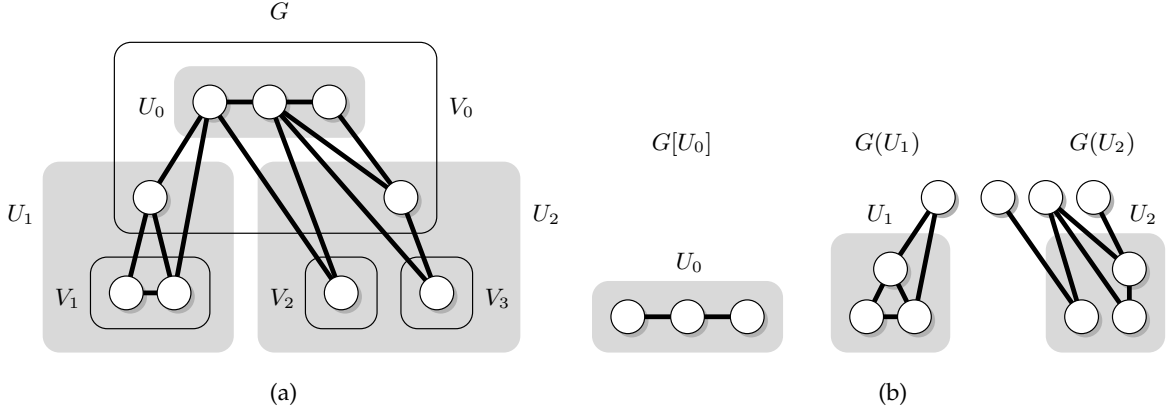}}
\caption{(a) Separation $[V_0,V_1,V_2,V_3]$ is a refinement of
  separation $[U_0,U_1,U_2]$. (b) Subgraphs $G[U_0]$, $G(U_1)$, and
  $G(U_2)$ of $G$.}
\label{figure:refinement}
\end{figure}
Let $[V_0,\ldots,V_p]$ and $[U_0,\ldots,U_q]$ be two separations of graph $G$.
We say that $[V_0,\ldots,V_p]$ is a {\em
  refinement} of $[U_0,\ldots,U_q]$ if the following
three properties hold.
\begin{enumerate}[\em Property~R1:]
\addtolength{\itemsep}{-0.5\baselineskip}
\item 
\label{refine:1}
$U_0\subseteq V_0$.

\item 
\label{refine:2}
For each index $i=1,\ldots,p$, there is an index $j$ with $1\leq
j\leq q$ and $V_i\subseteq U_j$.

\item
\label{refine:3}
For any indices $i$, $i'$, $i''$ with $1\leq i<i'<i''\leq p$, if
$V_i\cup V_{i''}\subseteq U_j$, then $V_{i'}\subseteq U_j$.

\end{enumerate}
For instance, in Figure~\ref{figure:refinement}(a), 
$[V_0,V_1,V_2,V_3]$ is a refinement of $[U_0,U_1,U_2]$.
Below is the main lemma of the section.

\begin{lemma}
\label{lemma:lemma3.1}
Let $k$ be a positive integer.  Let $G$ be an $n$-node connected graph
embedded on a genus-$o(n/\ell_k^2)$ surface.  Given a
$(k-1)$-separation $\SSS_{k-1}$ of $G$, it takes $O(n)$ time to
compute a $k$-separation $\SSS_{k}$ of $G$ that is a refinement of
$\SSS_{k-1}$.
\end{lemma}

The proof of Lemma~\ref{lemma:lemma3.1} needs the following lemma, which
can be proved by Lemmas~\ref{lemma:goodrich} and~\ref{lemma:djidjev}.

\begin{lemma}
\label{lemma:lemma2.1}
Let $k$ be a positive integer.  Given an $n$-node graph $G$ embedded
on a genus-$o(n/\ell_k^2)$ surface, it takes $O(n)$ time to compute an
$o(\frac{n}{\ell_k})$-node subset $V$ of $\V(G)$ such that each node
of $\V(G)\setminus V$ has degree at most $\ell_k^2$ in $G$ and each
connected component of $G\setminus V$ has at most $\ell_k^4$ nodes.
\end{lemma}

\begin{proof}
We first apply Lemma~\ref{lemma:djidjev} to compute in $O(n)$ time an
$o(\frac{n}{\ell_k})$-node subset $V'$ of $\V(G)$ such that
$G\setminus V'$ is planar.  We then apply Lemma~\ref{lemma:goodrich}
to compute in $O(n)$ time a separator decomposition $\T$ of
$G\setminus V'$.  For each vertex $S$ of $\T$, let $\D(S)$ denote the
union of all the vertices in the subtree of $\T$ rooted at $S$ and let
$\dd(S)=|\D(S)|$.  Let $r=\ell_k^2$.  Let $V''$ consist of the nodes
of $G$ with degree more than $r$ in $G$.  Let $V'''$ be the union of
all the vertices $S$ of $\T$ with $\dd(S)>r^2$.  Let $V=V'\cup V''\cup
V'''$.  By $V'\cup V'''\subseteq V$ and the definition of $\T$, each
connected component of $G\setminus V$ has at most $r^2$ nodes.  By
$V''\subseteq V$, each node of $\V(G)\setminus V$ has degree at most
$r$ in $G$.  Since $G$ has $O(n)$ edges,
$|V''|=O(\frac{n}{r})=o(\frac{n}{\ell_k})$.
It remains to show $|V'''|=o(\frac{n}{\ell_k})$.  For each index
$i\geq 1$, let $\II_i$ consist of the vertices $S$ of $\T$ with
$r^2\cdot (\frac{3}{2})^{i-1}<\dd(S)\leq r^2\cdot (\frac{3}{2})^i$.
By $r^2\geq 1$ and $i\geq 1$, each $S\in\II_i$ is an internal vertex
of $\T$.  By definition of $\T$, we know that $\D(S)$ and $\D(S')$ are
disjoint for any two distinct elements $S$ and $S'$ of $\II_i$,
implying that $\sum_{S\in \II_i}\dd(S)\leq n$ holds.  Since
$\dd(S)>r^2\cdot(1.5)^{i-1}$ holds for each $S\in \II_i$, we have
$|\II_i| < \frac{n}{r^2\cdot (1.5)^{i-1}}$.  Since each $S\in \II_i$
is an internal vertex of $\T$, $S$ is a separator of $G[\D(S)]$.
Therefore, $|S|=O(r\cdot (1.5)^{i/2})$ holds for each vertex $S$ in
$\II_i$.  We have $|V'''|=\sum_{i\geq 1}\sum_{S\in
  \II_i}|S|=\sum_{i\geq1}O(\frac{n}{r\cdot(1.5)^{i/2}})=O(\frac{n}{r})=o(\frac{n}{\ell_k})$.
The lemma is proved.
\end{proof}

\begin{algorithm}[t]
{
\begin{minipage}[t]{1cm}
\begin{tabbing}
\quad\=\quad\=\quad\=\quad\=\quad\=\quad\=\quad\=\kill
Let $p=0$ and let all elements of $\C$ be initially unmarked.\\
For each $j=1,\ldots,q$, perform the following repeat-loop.\\
\>Repeat the following steps until all elements of $\C_j$ are marked.\\
\>\>Let $v_0$ be an arbitrary node of $V_0$ that is adjacent to some
unmarked element of $\C_j$.\\
\>\>Let $\U$ consist of the unmarked elements of $\C_j$ that are
adjacent to $v_0$ in $G$.\\
\>\>Let $C_{i_1},\ldots,C_{i_3}$ be the elements
of $\U$ in clockwise order around $v_0$ in $G$.\\
\>\>Mark all $i_3-i_1+1$ elements of $\U$.\\
\>\>Repeat the following four steps until $i_1>i_3$.\\
\>\>\>Let $i_2$ be the largest index with $i_1\leq i_2\leq i_3$ and
$|C_{i_1}|+\cdots+|C_{i_2}|\leq \ell_k^4$.\\
\>\>\>Let $p=p+1$.\\
\>\>\>Let $\xx_p = v_0$ and $V_p=C_{i_1}\cup \cdots \cup C_{i_2}$.\\
\>\>\>Let $i_1=i_2+1$.\\
Output $V_1,\ldots,V_p$ and $\textit{hook}_1,\ldots,\textit{hook}_p$.
\end{tabbing}
\end{minipage}
}
\caption{}
\label{algorithm:clustering}
\end{algorithm}

\begin{figure}[t]
\centerline{\input{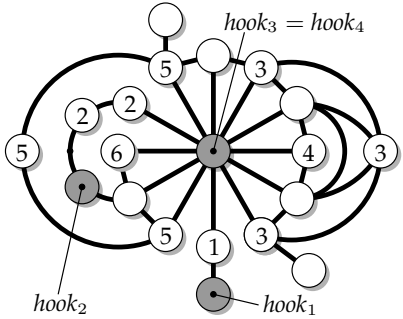}}
\caption{An illustration for Algorithm~\ref{algorithm:clustering}.}
\label{figure:clustering}
\end{figure}

\begin{proof}[Proof of Lemma~\ref{lemma:lemma3.1}]
Suppose that $[U_0,\ldots,U_q]$ is the given $(k-1)$-separation
$\SSS_{k-1}$.  Let $V'_0$ be the $O(n)$-time computable subset of
$\V(G)$ ensured by Lemma~\ref{lemma:lemma2.1}.  We have
$|V'_0|=o(\frac{n}{\ell_k})$.  Let $V_0=U_0\cup V'_0$.  Let $\C$
consist of the connected components of $G\setminus V_0$.  By
$V'_0\subseteq V_0$, each element of $\C$ has at most $\ell_k^4$
nodes.  By $U_0\subseteq V_0$ and Properties~S\ref{separation:1}
and~S\ref{separation:2} of $\SSS_{k-1}$, each element of $\C$ is
contained by some $U_j$ with $1\leq j\leq q$.  For each
$j=1,\ldots,q$, let $\C_j$ consist of the elements $C$ of $\C$ with
$C\subseteq U_j$.  We run Algorithm~\ref{algorithm:clustering} to
obtain (a) a disjoint partition $V_1,\ldots,V_p$ of $G\setminus V_0$
and (b) $p$ nodes $\xx_1,\ldots,\xx_p$ of $V_0$, which may not be
distinct. Let $\SSS_k = [V_0,\ldots,V_p]$.  Since $G$ is connected,
each element of $\C$ is adjacent to $V_0$.  The first statement of the
outer repeat-loop is well defined.  Since each element of $\C$ has at
most $\ell_k^4$ nodes, the first statement of the inner repeat-loop is
well defined.  See Figure~\ref{figure:clustering} for an illustration:
Suppose that all nodes are in $U_1$.  All nodes are initially
unmarked.  Let $V_0$ consist of the nine unlabeled nodes, including
the three gray nodes.  For each $i=1,\ldots,6$, let $C_i$ consist of
the nodes with label $i$. That is, $C_1,\ldots,C_6$ are the six
connected components of $G\setminus V_0$. Suppose that $\ell_k^4=7$
and the first two iterations of the outer repeat-loop obtain $V_1=C_1$
and $V_2=C_2$. In the third iteration of the outer repeat-loop,
$C_3,\ldots,C_6$ are the unmarked elements of $\C$ that are adjacent
to $\textit{hook}_3$ in clockwise order around $\textit{hook}_3$. By
$|C_3|+|C_4|+|C_5|=7$, the two iterations of the inner repeat-loop
obtain $V_3=C_3\cup C_4\cup C_5$ and $V_4=C_6$.

By definition of Algorithm~\ref{algorithm:clustering}, one can verify
that Properties~R\ref{refine:1},~R\ref{refine:2}, and~R\ref{refine:3}
hold for $\SSS_{k-1}$ and $\SSS_k$ (that is, $\SSS_k$ is a refinement
of $\SSS_{k-1}$) and Properties~S\ref{separation:1}
and~S\ref{separation:2} hold for $\SSS_k$.  By
Property~S\ref{separation:3} of~$\SSS_{k-1}$, we have
$|U_0|=o(\frac{n}{\ell_{k-1}})=o(\frac{n}{\ell_k})$.  By
$|V'_0|=o(\frac{n}{\ell_k})$, we have $|V_0|\leq
|U_0|+|V'_0|=o(\frac{n}{\ell_k})$.  Let $I_{\text{small}}$ consist of
the indices $i$ with $1\leq i\leq p$ and $|V_i| \leq
\frac{1}{2}\cdot\ell_k^4$.  Let $I_{\text{large}}$ consist of the
indices $i$ with $1\leq i\leq p$ and $|V_i| >
\frac{1}{2}\cdot\ell_k^4$.  We show
$p=|I_{\text{small}}|+|I_{\text{large}}|=o(\frac{n}{\ell_k})$ as
follows.  By Property~S\ref{separation:1} of $\SSS_k$, we have
$|I_{\text{large}}|= o(\frac{n}{\ell_k})$.  To show
$|I_{\text{small}}|= o(\frac{n}{\ell_k})$, we categorize the indices
$i$ in $I_{\text{small}}$ with $1\leq i<p$ into the the following
types, where $j$ is the index with $V_i\subseteq U_j$:
\begin{description}
\addtolength{\itemsep}{-0.5\baselineskip}
\item[\rm\em Type 1:] $i\in I_{\text{small}}$ and $i+1\in
  I_{\text{large}}$.  The number of such indices $i$ is at
  most $|I_{\text{large}}|=o(\frac{n}{\ell_k})$.

\item[\rm\em Type 2:] $i\in I_{\text{small}}$ and $i+1\in
  I_{\text{small}}$.
\begin{description}
\addtolength{\itemsep}{-0.2\baselineskip}
\item[\rm\em Type 2a:] $V_{i+1}\subseteq U_{j+1}$.  The number of such
  indices $i$ is at most
  $q=o(\frac{n}{\ell_{k-1}})=o(\frac{n}{\ell_k})$.

\item[\rm\em Type 2b:] $V_{i+1}\subseteq U_j$ and
$\textit{hook}_i\in V_0\setminus U_0$.  By
  Properties~S\ref{separation:1} and~S\ref{separation:2} of
  $\SSS_{k-1}$, we know that $\textit{hook}_i\in U_j$. By definition
  of Algorithm~\ref{algorithm:clustering},
  $\textit{hook}_{i'}\ne\textit{hook}_i$ holds for all indices $i'$
  with $i<i'\leq p$.  The number of such indices $i$ is at
  most $|V_0\setminus U_0|\leq |V_0|=o(\frac{n}{\ell_k})$.

\item[\rm\em Type 2c:] $V_{i+1}\subseteq U_j$ and $\textit{hook}_i\in
  U_0$. We have $\textit{hook}_i\in \N_G(U_j)$.  By definition of
  Algorithm~\ref{algorithm:clustering},
  $\textit{hook}_{i'}\ne\textit{hook}_i$ holds for all indices $i'>i$
  with $V_{i'}\subseteq U_j$.  By Property~S\ref{separation:5} of
  $\SSS_{k-1}$, the number of such indices $i$ is at most
  $\sum_{j=1}^q\nn_G(U_j)=o(\frac{n}{\ell_{k-1}})=o(\frac{n}{\ell_k})$.
\end{description}
\end{description}
We have $p=o(\frac{n}{\ell_k})$.  Property~S\ref{separation:3} holds
for $\SSS_k$.  By definition of Algorithm~\ref{algorithm:clustering},
$|V_i|\leq \ell_k^4$ holds for each $i=1,\ldots,p$.  By $V'_0\subseteq
V_0$, each node of $\V(G)\setminus V_0$ has degree at most
$\ell_k^2$. Property~S\ref{separation:4} holds for $\SSS_k$.

\begin{figure}[t]
\centerline{\input{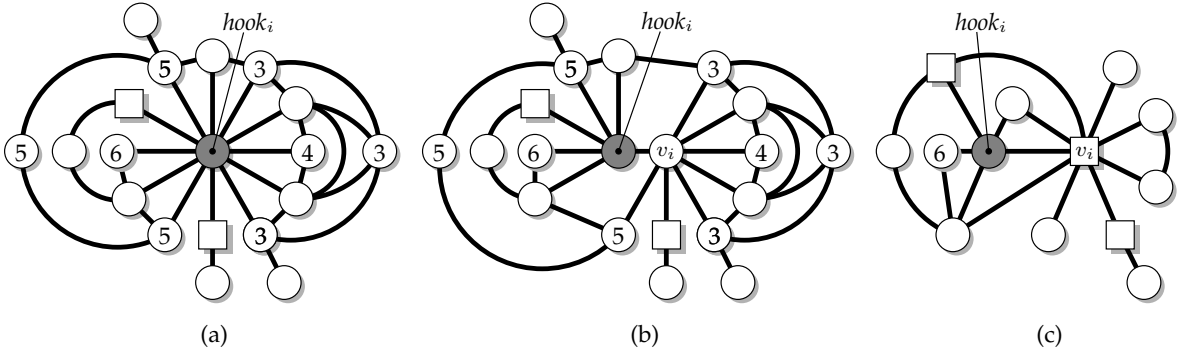}}
\caption{The operation that contracts all nodes of $V_i$ into a node
  $v_i$, which takes over some neighbors of $\xx_i$.}
\label{figure:operation}
\end{figure}

To see Property~S\ref{separation:5} of $\SSS_k$, we obtain a
contracted graph from $G$ by performing the following two steps for
each $i=1,\ldots,p$.\footnote{The contraction procedure is only for
  proving Property~S\ref{separation:5} of $\SSS_k$, not needed for
  computing $\SSS_k$.}
{\em Step~1:} Let $C_{i_1},\ldots,C_{i_2}$ be the elements of $\C$
with $V_i=C_{i_1}\cup C_{i_1+1}\cup \cdots \cup C_{i_2}$ in clockwise
order around $\xx_i$ in $G$.  Split $\xx_i$ into two adjacent nodes
$\xx_i$ and $v_i$ and let $v_i$ take over the neighbors of $\xx_i$ in
clockwise order around $\xx_i$ from the first neighbor of $\xx_i$ in
$C_{i_1}$ to the first neighbor of $\xx_i$ in $C_{i_2}$. {\em Step~2:}
Contract all nodes of $V_i$ into node $v_i$ and delete multiple edges
and self-loops.  See Figure~\ref{figure:operation} for an
illustration: For each $i=3,\ldots,6$, let $C_i$ consist of the nodes
with labels~$i$ in Figure~\ref{figure:operation}(a).  Suppose that
$i_1=3$, $i_2=5$, and $V_i=C_3\cup C_4\cup C_5$. The unlabeled circle
nodes belong to $V_0$. The square nodes are two previously contracted
nodes $v_{i'}$ and $v_{i''}$ from $V_{i'}$ and $V_{i''}$ for some
indices $i'$ and $i''$ with $1\leq i'\ne i''<i$.
Figure~\ref{figure:operation}(b) shows the result of
Step~1. Figure~\ref{figure:operation}(c) shows the result of Step~2.
Observe that each node that is adjacent to $V_i$ becomes a neighbor of
$v_i$ after applying Steps~1 and~2.  Also, each neighbor of $\xx_i$
that is not in $V_i$ either remains a neighbor of $\xx_i$ or becomes a
neighbor of $v_i$ after applying Steps~1 and~2.  Therefore, for each
$i=1,\ldots,p$ and each node $v_0\in \N_G(V_i)$, there is either an
edge $(v_0,v_i)$ or an edge $(v_i,v_{i'})$ for some index $i'$ with
$i'>i$ and $\xx_{i'}=v_0$.  Thus, $\sum_{i=1}^p \nn_G(V_i)$ is no more
than the number of edges in the resulting contracted simple graph,
which has $|V_0|+p=o(\frac{n}{\ell_k})$ nodes.  Observe that Step~1
does not increase the genus of the embedding. Since the subgraph
induced by $V_i\cup \{v_i\}$ is connected, Step~2 does not increase
the genus of the embedding, either.  The number of edges in the
resulting contracted simple genus-$o(n/\ell_k^2)$ graph is
$o(\frac{n}{\ell_k})$.  Property~S\ref{separation:5} holds for
$\SSS_k$.  The lemma is proved.
\end{proof}

\section{Our compression scheme}
\label{section:compression}
This section proves Theorem~\ref{theorem:theorem1}.

\subsection{Recovery string}

A {\em labeling} of graph $G$ is a one-to-one mapping from $\V(G)$ to
$\{0,1,\ldots,\vv(G)-1\}$. For instance,
Figure~\ref{figure:labeling}(a) shows a labeling for graph $G$.  Let
$G$ be a graph embedded on a surface.  We say that a graph $\Delta$
embedded on the same surface is a {\em triangulation} of $G$ if $G$ is
a subgraph of $\Delta$ with $\V(\Delta)=\V(G)$ such that each face of
$\Delta$ has three nodes.
The following lemma shows an $o(n)$-bit string with which the larger
embedded labeled subgraphs of $G$ can be recovered from smaller
embedded labeled subgraphs of $G$ in $O(n)$ time.

\begin{lemma}
\label{lemma:recovery}
Let $k$ be a positive integer.  Let $G$ be an $n$-node graph embedded
on a genus-$o(\frac{n}{\ell_k})$ surface.  Let $\Delta$ be a
triangulation of $G$.  Let $\SSS_k=[V_0,\ldots,V_p]$ be a given
$k$-separation of $\Delta$ and $\SSS_{k-1}=[U_0,\ldots,U_q]$ be a
given $(k-1)$-separation of $\Delta$ such that $\SSS_k$ is a
refinement of $\SSS_{k-1}$.  For any given labeling $L_{k,i}$ of
$G(V_i)$ for each $i=1,\ldots,p$, the following statements hold.
\begin{enumerate}
\addtolength{\itemsep}{-0.5\baselineskip}
\item 
\label{rec1}
It takes overall $O(n)$ time to compute a labeling $L_{k-1,j}$ of
subgraph $G(U_j)$ for each $j=1,\ldots,q$.

\item 
\label{rec2}
Given the above labelings $L_{k-1,j}$ of subgraphs $G(U_j)$ with
$1\leq j\leq q$, it takes $O(n)$ time to compute an $o(n)$-bit string
$\textit{Rec}_k$ such that $G(U_j)$ and $L_{k-1,j}$ for all
$j=1,\ldots,q$ can be recovered in overall $O(n)$ time from
$\textit{Rec}_k$ and $G(V_i)$ and $L_{k,i}$ for all $i=1,\ldots,p$.
\end{enumerate}
\end{lemma}

\begin{figure}[t]
\centerline{\input{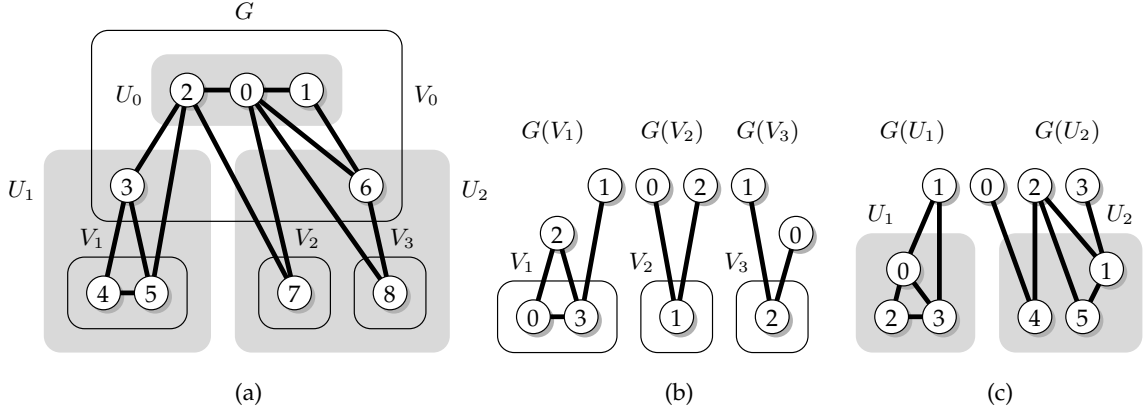}}
\caption{(a) Graph $G$ with a labeling.  (b) Subgraphs $G(V_1)$,
  $G(V_2)$, and $G(V_3)$ of $G$ with labelings.  (c) Subgraphs
  $G(U_1)$ and $G(U_2)$ of $G$ with labelings.}
\label{figure:labeling}
\end{figure}

\begin{proof}
Since $\Delta$ is a subgraph $G$ with $\V(\Delta)=\V(G)$, one can
easily verify that $\SSS_{k-1}$ (respectively, $\SSS_k$) is also a
$(k-1)$-separation (respectively, $k$-separation) of $G$.
For each $j=1,\ldots,q$, let $I_j$ consist of the indices $i$ with
$V_i\subseteq U_j$. Let $W_j$ consist of the nodes of $G(U_j)$ that
are not in any $V_i$ with $i\in I_j$. By
Properties~S\ref{separation:1} and~S\ref{separation:2} of $\SSS_k$,
$W_j\subseteq V_0$.  For instance, if $G$ is as shown in
Figure~\ref{figure:labeling}(a), where $v_t$ with $0\leq t\leq 8$
denotes the node with label~$t$.  We have
$I_1=\{1\}$, $I_2=\{2,3\}$, $W_1=\{v_2,v_3\}$, and
$W_2=\{v_0,v_1,v_2,v_6\}$.  Let the labeling $L_{k-1,j}$ for $G(U_j)$
be defined as follows.
\begin{itemize}
\addtolength{\itemsep}{-0.5\baselineskip}
\item 
For the nodes of $G(U_j)$ in $W_j$, let $L_{k-1,j}$ be an arbitrary
one-to-one mapping from $W_j$ to $\{0,1,\ldots,|W_j|-1\}$. In
Figure~\ref{figure:labeling}(c), we have $L_{k-1,1}(v_2)=1$,
$L_{k-1,1}(v_3)=0$, $L_{k-1,2}(v_0)=2$, $L_{k-1,2}(v_1)=3$,
$L_{k-1,2}(v_2)=0$, and $L_{k-1,2}(v_6)=1$.

\item 
For the nodes of $G(U_j)$ not in $W_j$, let $L_{k-1,j}$ be the
one-to-one mapping from $\bigcup_{i\in I_j} V_i$ to $\{|W_j|, |W_j|+1,
\ldots, \vv(G(U_j))-1\}$ obtained by sorting $(i,L_{k,i}(v))$ for all
indices $i\in I_j$ and all nodes $v\in V_i$ such that
$L_{k-1,j}(v)<L_{k-1,j}(v')$ holds for a node $v$ of $V_i$ and a node
$v'$ of $V_{i'}$ if and only if (a) $i<i'$ or (b) $i=i'$ and
$L_{k,i}(v)<L_{k,i'}(v')$.  For instance, if $L_{k,1}$, $L_{k,2}$, and
$L_{k,3}$ are as shown in Figure~\ref{figure:labeling}(b), then
$L_{k-1,1}$ and $L_{k-1,2}$ can be as shown in
Figure~\ref{figure:labeling}(c) and $L_{k-2,1}$ can be as shown in
Figure~\ref{figure:labeling}(a).
\end{itemize}
It takes $O(\vv(G(U_j)))=O(|U_j|+\nn_G(U_j))$ time to compute
$L_{k-1,j}$ from all $L_{k,i}$ with $i\in I_j$. By
Property~S\ref{separation:5} of~$\SSS_{k-1}$, it takes overall $O(n)$
time to compute all $L_{k-1,j}$ with $1\leq j\leq q$ from all
$L_{k,i}$ with $1\leq i\leq p$.
Statement~\ref{rec1} is proved.

By Property~S\ref{separation:4} of $\SSS_{k-1}$, the label of each
node of $G(U_j)$ assigned by $L_{k-1,j}$ can be represented by
$O(\log\text{poly}(\ell_{k-1}))=O(\ell_k)$ bits.
By Property~S\ref{separation:4} of $\SSS_{k}$, the label of each node
of $G(V_i)$ assigned by $L_{k,i}$ can be represented by
$O(\log\text{poly}(\ell_k))=O(\ell_{k+1})$ bits.
For each index $j=1,\ldots,q$, 
\begin{itemize}
\addtolength{\itemsep}{-0.5\baselineskip}
\item 
string $\textit{Rec}'_{k,j}$ stores the adjacency list of the
embedded subgraph of $G(V_j)$ induced by $W_j$ via the labeling
$L_{k-1,j}$ of $W_j$,
\item 
string $\textit{Rec}''_{k,j}$ stores the information required to
recover $L_{k-1,j}$ from all $L_{k,i}$ with $i\in I_j$, and
\item 
string $\textit{Rec}'''_{k,j}$ stores the information required to
recover the embedding of $G(U_j)$ from the embeddings of all $G(V_i)$
with $i\in I_j$ and the embedding of the subgraph of $G(U_j)$ induced
by $W_j$.
\end{itemize}
By definition of $W_j$, we have $|W_j| = |V_0\cap U_j|+\nn_G(U_j)$. It
follows from Property~S\ref{separation:3} of $\SSS_k$ and
Property~S\ref{separation:5} of $\SSS_{k-1}$ that 
\begin{displaymath}
\sum_{j=1}^q |W_j|\leq |V_0|+\sum_{j=1}^q \nn_G(U_j)=
o\left(\frac{n}{\ell_k}\right)+o\left(\frac{n}{\ell_{k-1}}\right)=o\left(\frac{n}{\ell_k}\right).
\end{displaymath}
Let $W=\bigcup_{j=1}^q W_j$.  Since $G[V_0],G(V_1),\ldots,G(V_p)$ form
a disjoint partition of the edges of $G$, the overall number of edges
in the subgraphs of $G(V_j)$ induced by $W_j$ for all $j=1,\ldots,q$
is no more than the number of edges in $G[W]$, which is
$O(|W|+o(\frac{n}{\ell_k}))\leq O(\sum_{j=1}^q
|W_j|)+o(\frac{n}{\ell_k})=o(\frac{n}{\ell_k})$.  Therefore,
\begin{equation}
\label{eq:1}
\sum_{j=1}^q\|\textit{Rec}'_{k,j}\|
=o\left(\frac{n}{\ell_k}\right)\cdot O(\ell_k) =o(n).
\end{equation}
It suffices for $\textit{Rec}''_{k,j}$ to store the list of
$(i,L_{k,i}(v),L_{k-1,j}(v))$ for all $i\in I_j$ and all $v\in
\N_G(V_i)$.  By Property~R\ref{refine:3} of $\SSS_{k-1}$ and $\SSS_k$
and Property~S\ref{separation:4} of $\SSS_{k-1}$, index $i$ can be
represented by an $O(\ell_k)$-bit offset $t$ such that $i$ is the
$t$-th smallest index in $I_j$.  Thus,
$\|\textit{Rec}''_{k,j}\|=\sum_{i\in I_j} \nn_G(V_i)\cdot O(\ell_k)$.
By Property~S\ref{separation:5} of $\SSS_k$, we have $\sum_{j=1}^q
\sum_{i\in I_j} \nn_G(V_i)= \sum_{i=1}^p \nn_G(V_i)=
o(\frac{n}{\ell_k})$.  Therefore,
\begin{equation}
\label{eq:2}
\sum_{j=1}^q \|\textit{Rec}''_{k,j}\|
=o\left(\frac{n}{\ell_k}\right)\cdot O(\ell_k) = o(n).
\end{equation}
It suffices for $\textit{Rec}'''_{k,j}$ to store the list of
$(L_{k-1,j}(v),L_{k-1,j}(v'),L_{k-1,j}(v''))$ for all pairs of edges
$(v,v')$ and $(v,v'')$ of $G(U_j)$ such that (a) $v''$ is the neighbor
of $v$ that immediately succeeds $v'$ in clockwise order around $v$ in
$G(U_j)$ and (b) nodes $v'$ and $v''$ are not in the same partition of
$\V(G(U_j))$ formed by the $|I_j|+1$ disjoint sets $W_j$ and $V_i$
with $i\in I_j$.  By Property~S\ref{separation:2} of $\SSS_k$, node
$v$ belongs to $W_j\subseteq V_0$. Since $\Delta$ is a triangulation
of $G$, the neighbors of $v$ in $\Delta$ form a cycle that surrounds
$v$ in $\Delta$.  Let $P$ be the path of the cycle from $v'$ to $v''$
in clockwise order around $v$.  At least one node, say, $u$ of $P$
belongs to $V_0$, since otherwise Property~S\ref{separation:2} of
$\SSS_k$ would imply that all nodes of $P$ belong to the same $V_i$
for some $i\in I_j$, contradicting with the choices of $v'$ and $v''$.
Edge $(v,u)$ belongs to $\Delta[V_0]$.
Observe that each edge of $\Delta[V_0]$ can be identified by at most
four such edge pairs $(v,v')$ and $(v,v'')$.  Since the edges of
$G(U_j)$ and $G(U_{j'})$ with $1\leq j\ne j'\leq q$ are disjoint, the
number of edge pairs stored in $\textit{Rec}'''_{k,j}$ is at most four
times the number of edges in $\Delta[V_0]$. By
Property~S\ref{separation:3} of $\SSS_k$ and the fact that $\Delta$
has genus $o(\frac{n}{\ell_k})$, the number of edge pairs stored in
$\textit{Rec}'''$ is $o(\frac{n}{\ell_k})$.  Therefore,
\begin{equation}
\label{eq:3}
\sum_{j=1}^q \|\textit{Rec}'''_{k,j}\|
=o\left(\frac{n}{\ell_k}\right)\cdot O(\ell_k)=o(n).
\end{equation}
Let 
\begin{eqnarray*}
\textit{Rec}'_k&=&\textit{Rec}'_{k,1}\circ \cdots \circ\textit{Rec}'_{k,q},\\
\textit{Rec}''_k&=&\textit{Rec}''_{k,1}\circ \cdots \circ\textit{Rec}''_{k,q},\\
\textit{Rec}'''_k&=&\textit{Rec}'''_{k,1}\circ \cdots \circ\textit{Rec}'''_{k,q},\\
\textit{Rec}_k&=&\textit{Rec}'_k\circ\textit{Rec}''_k\circ\textit{Rec}'''_k.
\end{eqnarray*}
By Equations~(\ref{eq:1}),~(\ref{eq:2}), and~(\ref{eq:3}) and
Lemma~\ref{lemma:unary-concat}, we have $\|\textit{Rec}_k\|=o(n)$. It
takes $O(n)$ time to compute $\textit{Rec}_k$ from all labelings
$L_{k,j}$ and all embedded graphs $G(U_j)$ with $1\leq j\leq q$ and
all labelings $L_{k-1,i}$ and all embedded graphs $G(V_i)$ with $1\leq
i\leq p$.  It also takes $O(n)$ time to recover all labelings
$L_{k,j}$ and all embedded graphs $G(U_j)$ with $1\leq j\leq q$ from
$\textit{Rec}_k$ and all labelings $L_{k-1,i}$ and all embedded graphs
$G(V_i)$ with $1\leq i\leq p$.  Statement~\ref{rec2} holds.  The lemma
is proved.
\end{proof}

\subsection{Proving Theorem~\ref{theorem:theorem1}}
We are ready to prove the main theorem of the paper.
\begin{proof}[Proof of Theorem~\ref{theorem:theorem1}]\
%\paragraph{Condition~C\ref{condition:c1}}
%Let $\G$ be the hereditary class of graphs. 
Let $G\in \G$ be the $n$-node input graph embedded on a
genus-$o(\frac{n}{\log^2 n})$ surface.  The encoding algorithm
$\textit{Encode}_A$ performs the following four steps on $G$.
\begin{enumerate}[\em E1:]
\addtolength{\itemsep}{-0.5\baselineskip}

\item 
\label{e1}
Triangulate the embedded graph $G$ into a triangulation $\Delta$ of
$G$.  Let $\SSS_0$ be the $0$-separation $[\varnothing,\V(\Delta)]$ of
$\Delta$. For each $k=1,2$, apply Lemma~\ref{lemma:lemma3.1} to
obtain a $k$-separation $\SSS_k$ of $\Delta$ that is a refinement of
$\SSS_{k-1}$.

\item 
\label{e2}
Let $[V_0,\ldots,V_p]=\SSS_2$.  Apply Lemma~\ref{lemma:basis} with
$\ell=\max_{1\leq i\leq p} \vv(G(V_i))$ to compute
(a)~$\textit{Label}(H)$ and $\textit{Optcode}(H)$ for all distinct
graphs $H$ in class $\G$ with $\vv(H)\leq\ell$ and
(b)~$\textit{Table}(\G,\ell)$.  For each $i=1,\ldots,p$, apply
Lemma~\ref{lemma:basis}(\ref{basis:1}) to compute from
$\textit{Table}(\G,\ell)$ the binary string
$\textit{Code}(V_i)=\textit{Optcode}(G(V_i))$ and the labeling
$L_{2,i}=\textit{Label}(G(V_i))$.

\item 
\label{e3}
For each $k=2,1$, perform the following two substeps.
\begin{enumerate}[\em E\ref{e3}.1:]
\item 
\label{encode:concat}
Let $[U_0,\ldots,U_q]=\SSS_{k-1}$ and $[V_0,\ldots,V_p]=\SSS_k$.  For
each $j=1,\ldots,q$, let binary string
$\textit{Code}(U_j)=\textit{Code}(V_{i_1})\circ \cdots \circ
\textit{Code}(V_{i_2})$, where $\{i_1,i_1+1,\ldots,i_2\}$ are the
indices $i$ with $V_i\subseteq U_j$.

\item 
Apply Lemma~\ref{lemma:recovery}(\ref{rec1}) to obtain the labelings
$L_{k-1,j}$ of subgraphs $G(U_j)$ for all $j=1,\ldots,q$.  Apply
Lemma~\ref{lemma:recovery}(\ref{rec2}) to obtain the $o(n)$-bit binary
string $\textit{Rec}_k$.
\end{enumerate}

\item 
\label{e4}
By $\SSS_0=[\varnothing,\V(G)]$, now we have $\textit{Code}(\V(G))$ (and a
labeling $L_{0,1}$ for $G=G(\V(G))$).  The output binary string
$\textit{Code}_A(G)=
\textit{Code}(\V(G))\circ\textit{Table}(\G,\ell)\circ\textit{Rec}_1\circ
\textit{Rec}_2$.
\end{enumerate}
By Lemma~\ref{lemma:lemma3.1}, Step~E\ref{e1} takes $O(n)$ time.  By
Property~S\ref{separation:5} of $\SSS_2$, we have $\sum_{i=1}^p
\vv(G(V_i))=n+o(n)$.  By Lemma~\ref{lemma:basis}, Step~E\ref{e2} takes
$O(n)$ time.  By Lemmas~\ref{lemma:unary-concat}
and~\ref{lemma:recovery}, Step~E\ref{e3} takes $O(n)$ time.  By
Lemma~\ref{lemma:unary-concat}, Step~E\ref{e4} takes $O(n)$ time.
Therefore, the encoding algorithm $\textit{Encode}_A(G)$ runs in
$O(n)$ time. Condition~C\ref{condition:c1} holds.

\bigskip
\noindent
The decoding algorithm $\textit{Decode}_A$ performs the following five
steps on $\textit{Code}_A(G)$.
\begin{enumerate}[\em D1:]
\addtolength{\itemsep}{-0.5\baselineskip}
\item 
\label{d1}
Obtain $\textit{Code}(\V(G))$, $\textit{Table}(\G,\ell)$,
$\textit{Rec}_1$, and $\textit{Rec}_2$ from $\textit{Code}_A(G)$.

\item 
\label{d2}
Let $\SSS_0=[\varnothing,\V(G)]$.  For each $k=1,2$, perform the
following substep.
\begin{enumerate}[\em D2.1:]
\item 
Let $[U_0,\ldots,U_q]=\SSS_{k-1}$ and $[V_0,\ldots,V_p]=\SSS_k$.  For
each $j=1,\ldots,q$, obtain all $\textit{Code}(V_i)$ with
$V_i\subseteq U_j$ from $\textit{Code}(U_j)$.
\end{enumerate}

\item 
\label{d3}
Let $[V_0,\ldots,V_p]=\SSS_2$.  
Apply Lemma~\ref{lemma:basis}(\ref{basis:2}) to obtain $G(V_i)$ and
$L_{2,i}=\textit{Label}(G(V_i))$ from
$\textit{Code}(V_i)=\textit{Optcode}(G(V_i))$ and
$\textit{Table}(\G,\ell)$ for each $i=1,\ldots,p$.

\item 
\label{d4}
For each $k=2,1$, perform the following substep.
\begin{enumerate}[\em D4.1:]
\item 
Let $[U_0,\ldots,U_q]=\SSS_{k-1}$ and $[V_0,\ldots,V_p]=\SSS_k$.
Apply Lemma~\ref{lemma:recovery}(\ref{rec2}) to recover $G(U_j)$ and
$L_{k-1,j}$ with $1\leq j\leq q$ from $G(V_i)$ and $L_{k,i}$ with
$1\leq i\leq p$ and $\textit{Rec}_k$.
\end{enumerate}

\item 
\label{d5}
Output $G=G(\V(G))$.
\end{enumerate}
By Lemma~\ref{lemma:unary-concat}, Step~D\ref{d1} takes $O(n)$ time.
By Lemma~\ref{lemma:unary-concat}, Step~D\ref{d2} takes $O(n)$ time.
By
Property~S\ref{separation:5} of $\SSS_2$, we have $\sum_{i=1}^p
\vv(G(V_i))=n+o(n)$.  
By Lemma~\ref{lemma:basis}(\ref{basis:2}), Step~D\ref{d3} takes $O(n)$
time.
By Lemma~\ref{lemma:recovery}(\ref{rec2}), Step~D\ref{d4} takes $O(n)$
time.  Therefore, the decoding algorithm $\textit{Decode}_A(G)$ runs
in $O(n)$ time.  Condition~C\ref{condition:c2} holds.  By
$\SSS_0=[\varnothing, \V(G)]$, graph $G=G(\V(G))$ is correctly
recovered from $\textit{Code}_A(G)$ at the end of Step~D4. Therefore,
$A$ is a compression scheme for $\G$.

\bigskip
\noindent
To show Condition~C\ref{condition:c3}, we first prove the following
claim for each $k=1,2$.
\begin{claim}
\label{claim:claim1}
Suppose that $[U_0,\ldots,U_q]=\SSS_{k-1}$ and
$[V_0,\ldots,V_p]=\SSS_k$.  If $\sum_{i=1}^p\|\textit{Code}(V_i)\|\leq
\beta n+o(n)$ and
$\|\textit{Code}(V_i)\|=\text{poly}(\ell_k)$ holds for each
$i=1,\ldots,p$, then $\sum_{j=1}^q\|\textit{Code}(U_j)\|\leq
\beta n+o(n)$ and
$\|\textit{Code}(U_j)\|=\text{poly}(\ell_{k-1})$ holds for each
$j=1,\ldots,q$.
\end{claim}
\begin{proof}[Proof of Claim~\ref{claim:claim1}]
For each $j=1,2,\ldots,q$, let $I_j$ consist of the indices $i$ with
$V_i\subseteq U_j$.  
By Property~S\ref{separation:4} of $\SSS_{k-1}$, we have $|I_j|\leq
|U_j|=\text{poly}(\ell_{k-1})$. Therefore, $\sum_{i\in
  I_j}\|\textit{Code}(V_i)\| = \text{poly}(\ell_{k-1})$, implying
$O(\log \sum_{i\in I_j}\|\textit{Code}(V_i)\|)= O(\ell_k)$. By
Property~S\ref{separation:3} of $\SSS_k$, $\sum_{j=1}^q |I_j| = p =
O(\frac{n}{\ell_k})$.  By Lemma~\ref{lemma:unary-concat}, we have
\begin{displaymath}
\sum_{j=1}^q \|\textit{Code}(U_j)\| =\sum_{j=1}^q O(|I_j|\cdot
\ell_k)+\sum_{i=1}^p \|\textit{Code}(V_i)\| \leq \beta n+o(n).
\end{displaymath}
We also have 
\begin{displaymath}
\|\textit{Code}(U_j)\| =O(|I_j|\cdot \ell_k)+ \sum_{i\in I_j}
\|\textit{Code}(V_i)\| =\text{poly}(\ell_{k-1}).
\end{displaymath}
The claim is proved.
\end{proof}

Let $[V_0,\ldots,V_p]=\SSS_2$. For each $i=1,\ldots,p$, let
$n_i=\vv(G(V_i))=|V_i|+\nn_G(V_i)$.  By Property~S\ref{separation:5}
of $\SSS_2$, we have $\sum_{i=1}^p n_i=n+o(n)$.  By Step~E\ref{e2}
of~$\textit{Encode}_A(G)$ and Lemma~\ref{lemma:basis}, we have
$\|\textit{Code}(V_i)\| = \|\textit{Optcode}(G(V_i))\| = \lceil\log
\textit{num}(\G,n_i)\rceil \leq \beta n_i+o(n_i)$.  By
Property~S\ref{separation:4} of~$\SSS_2$ and the assumption that
$\log\textit{num}(\G,n)=O(n)$,
\begin{equation}
\label{eq4}
\|\textit{Code}(V_i)\|=\text{poly}(\ell_2)
\end{equation}
holds for each $i=1,2,\ldots,p$.  
We have
\begin{equation}
\label{eq5}
\sum_{i=1}^p \|\textit{Code}(V_i)\| \leq\sum_{i=1}^p \left(\beta
n_i+o(n_i)\right)=\beta\cdot (n+o(n))+o(\beta n+o(n)) =\beta n+o(n).
\end{equation}
Combining Equations~(\ref{eq4}) and~(\ref{eq5}),
Claim~\ref{claim:claim1} for $k=2,1$, and $\SSS_0=[\varnothing,
  \V(G)]$, we have $\|\textit{Code}(\V(G))\|\leq \beta n+o(n)$.  By
Lemma~\ref{lemma:unary-concat} and
$\|\textit{Table}(\G,\ell)\|+\|\textit{Rec}_1\|+\|\textit{Rec}_2\|
=o(n)$, we have $\|\textit{Code}_A(G)\|\leq \beta n+o(n)$.
Condition~C\ref{condition:c3} holds.  The theorem is proved.
\end{proof}

\section{Extension}
\label{section:extend}

This section proves Theorem~\ref{theorem:theorem2}.  The only place in
our proof of Theorem~\ref{theorem:theorem1} requiring $\G$ to be
hereditary is Step~E\ref{e2}: We need $G(V_i)\in \G$ so that
$\textit{Optcode}(G(V_i))$ and $\textit{Label}(G(V_i))$ can be
obtained from $\textit{Table}(\G,\ell)$. For a non-hereditary class
$\G$, we can substitute $G(V_i)$ by a graph $H_i\in \G$ that is close
to $G(V_i)$ for each $i=1,2,\ldots,p$ as long as the overall number of
bits required to encode the overall difference between $G(V_i)$ and
$H_i$ is $o(n)$. The following corollary is an example of such an
extension.

\begin{corollary}
\label{corollary:corollary1}
Let $\G$ be a class of graphs satisfying $\log
\textit{num}(\G,n)=O(n)$ and that any input $n$-node graph $G\in \G$
to be encoded comes with a genus-$o(\frac{n}{\log^2 n})$ embedding.
If for any $2$-separation $[V_0,\ldots,V_p]$ of any graph $G\in \G$,
there exist graphs $H_1,\ldots,H_p$ in $\G$ such that
each $G(V_i)$ with $1\leq i\leq p$ can be obtained from $H_i$ by first
deleting $O(\nn_G(V_i))$ nodes (together with their incident edges)
and then updating (adding or deleting) $O(\nn_G(V_i))$ edges,
then $\G$ admits an optimal compression scheme.
\end{corollary}

\begin{proof}
We revise algorithm $\textit{Encode}_A$ by updating Steps~E\ref{e2}
and~E\ref{e4} as follows.
\begin{enumerate}[\em E1':]
\addtolength{\itemsep}{-0.5\baselineskip}
\setcounter{enumi}{1}
\item 
\label{step:e2-prime}
Let $[V_0,\ldots,V_p]=\SSS_2$.  Compute $H_1,\ldots,H_p$ from
$G(V_1),\ldots,G(V_p)$.  Apply Lemma~\ref{lemma:basis} with
$\ell=\max_{1\leq i\leq p} \vv(H_i)$ to compute (a)
$\textit{Label}(H)$ and $\textit{Optcode}(H)$ for each distinct graph
$H\in \G$ with $\vv(H)\leq\ell$ and (b) $\textit{Table}(\G,\ell)$.
Apply Lemma~\ref{lemma:basis}(\ref{basis:1}) to compute
$\textit{Code}(V_i)=\textit{Optcode}(H_i)$ and
$L'_{2,i}=\textit{Label}(H_i)$ from $\textit{Table}(\G,\ell)$ for all
indices $i=1,\ldots,p$.  Let $L_{2,i}$ be the labeling of $G(V_i)$
obtained from the labeling $L'_{2,i}$ of $H_i$ such that if $v$ and
$v'$ are two distinct nodes of $G(V_i)$ with
$L'_{2,i}(v)<L'_{2,i}(v')$, then we have $L_{2,i}(v)<L_{2,i}(v')$.
Let $\textit{Fix}_i$ be the binary string storing the difference
between $H_i$ and $G(V_i)$ via labeling $L'_{2,i}$.  Let
$\textit{Fix}=\textit{Fix}_1\circ\cdots\circ\textit{Fix}_p$.

\setcounter{enumi}{3}
\item 
\label{step:e4-prime}
By $\SSS_0=[\varnothing,\V(G)]$, now we have $\textit{Code}(\V(G))$
(and a labeling $L_{0,1}$ for $G=G(\V(G))$).  The output binary string
$\textit{Code}_A(G)$ for $G$ is
$\textit{Code}(\V(G))\circ\textit{Table}(\G,\ell)\circ\textit{Rec}_1\circ
\textit{Rec}_2\circ \textit{Fix}$.
\end{enumerate}
By $O(1)^{\text{poly}(\ell)}=o(n)$, it takes $o(n)$ time to compute an
$o(n)$-bit string $\textit{Table}'$ such that graphs $H_1,\ldots,H_p$
satisfying the above conditions can be obtained from
$G(V_1),\ldots,G(V_p)$ and $\textit{Table}'$ in $O(n)$ time.  By
Property~S\ref{separation:5} of $\SSS_2$ and the conditions of
$H_1,\ldots,H_p$, we have $\sum_{i=1}^p \vv(G(V_i))\leq \sum_{i=1}^p
\vv(H_i)=n+o(n)$.  By Lemmas~\ref{lemma:unary-concat}
and~\ref{lemma:basis}, Step~E\ref{step:e2-prime}' takes $O(n)$ time.
By Lemma~\ref{lemma:unary-concat}, Step~E\ref{step:e4-prime}' takes
$O(n)$ time.  Therefore, Condition~C\ref{condition:c1} holds for the
revised $\textit{Encode}_A$.
We revise algorithm $\textit{Decode}_A$ by updating Steps~D\ref{d1}
and~D\ref{d3} as follows.
\begin{enumerate}[\em D1':]
\addtolength{\itemsep}{-0.5\baselineskip}
\item 
Obtain $\textit{Code}(G)$, $\textit{Table}(\G,\ell)$,
$\textit{Rec}_1$, $\textit{Rec}_2$, and~$\textit{Fix}$ from
$\textit{Code}_A(G)$.

\setcounter{enumi}{2}
\item 
Let $[V_0,\ldots,V_p]=\SSS_2$.  Apply
Lemma~\ref{lemma:basis}(\ref{basis:2}) to obtain $H_i$ and
$L'_{2,i}=\textit{Label}(H_i)$ from
$\textit{Code}(V_i)=\textit{Optcode}(H_i)$ and
$\textit{Table}(\G,\ell)$ for each $i=1,\ldots,p$.  Apply
Lemma~\ref{lemma:unary-concat} to obtain all $\textit{Fix}_i$ with
$1\leq i\leq p$ from $\textit{Fix}$.  Obtain $G(V_i)$ and $L_{2,i}$
from $\textit{Fix}_i$, $H_i$, and $L'_{2,i}$ for all $i=1,2,\ldots,p$.
\end{enumerate}
Both revised steps take $O(n)$ time.  Condition~C\ref{condition:c2}
holds for the revised $\textit{Decode}_A$.  Subgraph $G(V_i)$ can be
obtained from $H_i$ by first deleting $O(\nn_G(V_i))$ nodes (and their
incident edges) and then updating $O(\nn_G(V_i))$ edges.  It follows
from Property~S\ref{separation:4} of $\SSS_2$ that
$\|\textit{Fix}_i\|=O(\nn_G(V_i)\cdot \ell_2)$.  By
Property~S\ref{separation:5} of $\SSS_2$, we have
$\sum_{i=1}^p\|\textit{Fix}_i\|=o(\frac{n}{\ell_2})\cdot
O(\ell_2)=o(n)$. By Lemma~\ref{lemma:unary-concat}, we have
$\|\textit{Fix}\|=o(n)$.  Condition~C\ref{condition:c3} holds the
revised $\textit{Code}_A(G)$.  The corollary is proved.
\end{proof}

We use 
Corollary~\ref{corollary:corollary1} to prove
Theorem~\ref{theorem:theorem2}.
\begin{proof}[Proof of Theorem~\ref{theorem:theorem2}]

\begin{figure}[t]
\centerline{\input{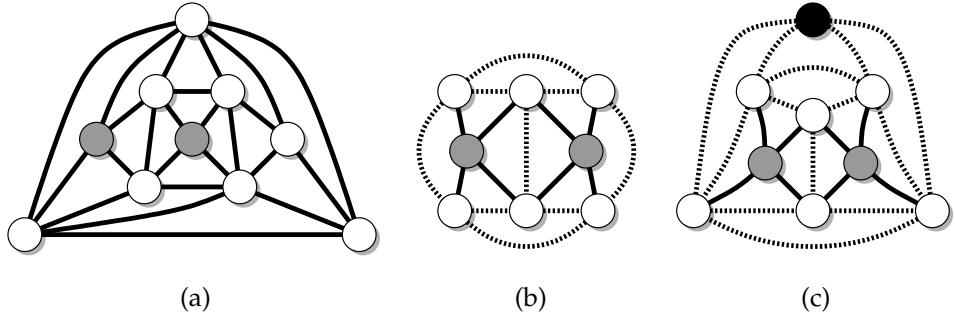}}
\caption{(a) A plane triangulation $\Delta$, where $V_i$ consists of
  the gray nodes.  (b) The subgraph $\Delta[V_i]$, where the dotted
  edges are those in $\Delta[V_i]\setminus \Delta(V_i)$. (c) A plane
  triangulation $H_i$ obtained by adding the dark node $v_F$ and four
  edges in the external face $F$ of $H_i$.}
\label{figure:triangulation-augment}
\end{figure}

Let $\Delta$ be an $n$-node triangulation of a genus-$g$ surface.  Let
$[V_0,\ldots,V_p]$ be a $2$-separation of $\Delta$.  Let $\F$ consist
of the non-triangle faces of $\Delta[V_i]$.  Let $H_i$ be the plane
triangulation obtained from $\Delta[V_i]$ by performing the following
two steps for each face $F\in \F$: (1) Add a node $v_F$ in $F$. (2)
For each node $u$ on the boundary of $F$, add an edge $(u,v_F)$.  See
Figure~\ref{figure:triangulation-augment} for an illustration.  Since
$\Delta$ is a triangulation, the boundary of $F$ contains at least two
nodes $u$ with $\N_\Delta(u)\not\subseteq \V(\Delta(V_i))$.
Therefore, at least two nodes of $\N(V_i)$ belong to the boundary of
$F$.  Let $e_F$ be an edge between two arbitrary nodes of $\N(V_i)$
that belong to the boundary of $F$.  The union of $e_F$ over all faces
$F\in \F$ has genus no more than $g=O(1)$. Therefore, the number of
added nodes to triangulate $\Delta[V_i]$ is $O(\nn_\Delta(V_i))$.  The
number of edges in $\Delta[V_i]\setminus \Delta(V_i)$ is also
$O(\nn_\Delta(V_i))$. Thus, $\Delta(V_i)$ can be obtained from $H_i$
by first deleting $O(\nn_\Delta(V_i))$ nodes together with their
incident edges and then deleting $O(\nn_\Delta(V_i))$ edges.  By
Corollary~\ref{corollary:corollary1},
Statement~\ref{thm2:triangulation} is proved.

\begin{figure}[t]
\centerline{\input{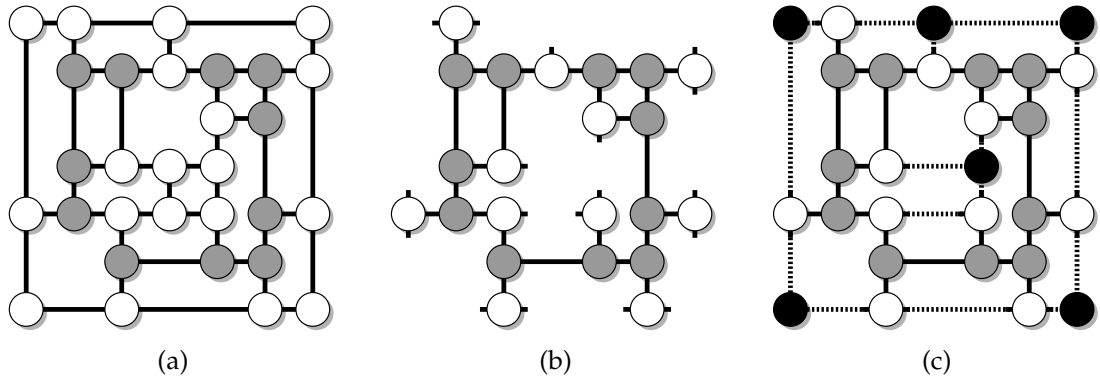}}
\caption{(a) A floorplan $G$, where $V_i$ consists of the gray nodes.
  (b) The subgraph $G(V_i)$. (c) A floorplan $H_i$ obtained from
  $G(V_i)$ by adding $O(\nn_G(V_i))$ nodes and edges.}
\label{figure:floorplan-augment}
\end{figure}

Let $G$ be an $n$-node floorplan.  Since each node of $G$ has at most
three neighbors in $G$, one can easily obtain a floorplan $H_i$ from
$G(V_i)$ by adding $O(\nn_G(V_i))$ nodes and edges. See
Figure~\ref{figure:floorplan-augment} for an example.
Statement~\ref{thm2:floorplan} follows from Corollary~\ref{corollary:corollary1}.
\end{proof}

\section{Concluding remarks}
\label{section:conclude}
Our optimal compression schemes rely on a linear-time obtainable
embedding. Can this requirement be dropped?  It would be of interest
to extend our compression schemes to support efficient queries and
updates.  We leave open the problems of obtaining optimal compression
schemes for $O(1)$-connected genus-$O(1)$ graphs and 3D
floorplans~\cite{CuestaAHPAM10,LiHZBYP06,SridharanDSXN09,LiMH09,WangYC09,WangZGGC08,LiHZBYPC06,CongM10}.

\small
\bibliographystyle{abbrv}
\bibliography{soda}
\end{document}